\documentclass[number,leqno]{elsarticle}
\usepackage{bbm}
\usepackage{color}
\usepackage{enumerate}
\usepackage{mathrsfs}
\usepackage{xspace}
\usepackage{graphicx}
\usepackage{ifthen}
\usepackage{amsmath,amssymb}
\usepackage{tikz}
\tikzstyle{every picture} = [>=latex]

\newcommand{\subexpsubexp}[1]{\ensuremath{\mbox{\sc DTIME}(2^{o(#1)})/\mbox{\sc SubEXP}}}
\newcommand{\PP}{\ensuremath{\mathrm{P}}\xspace}
\newcommand{\EXP}{\ensuremath{\mathrm{EXP}}\xspace}
\newcommand{\NEXP}{\ensuremath{\mathrm{NEXP}}\xspace}
\newcommand{\NP}{\ensuremath{\mathrm{NP}}\xspace}
\newcommand{\XP}{\ensuremath{\mathrm{XP}}\xspace}
\newcommand{\FPT}{\ensuremath{\mathrm{FPT}}\xspace}
\newcommand{\MCa}[2]{\ensuremath{\mathrm{MC}(#1,#2)}}
\newcommand{\cC}{{\mathcal C}}
\newcommand{\cD}{{\mathcal D}}
\newcommand{\cI}{{\mathcal I}}
\newcommand{\cG}{{\mathcal G}}
\newcommand{\cO}{{\mathcal O}}
\newcommand{\cP}{{\mathcal P}}
\newcommand{\cS}{{\mathcal S}}
\newcommand{\cV}{{\mathcal V}}
\newcommand{\ca}[1]{{\mathcal #1}}
\newcommand{\N}{\mathbb{N}}
\newcommand{\MC}[1]{\MCa{#1}{\cC}}
\newcommand{\MS}{\ensuremath{\mathrm{MSO}}\xspace}
\newcommand{\MSO}{\ensuremath{\mathrm{MSO}_1}\xspace}
\newcommand{\MSOii}{\ensuremath{\mathrm{MSO}_2}\xspace}
\newcommand{\MSOL}{\ensuremath{\mathrm{MSO_1\text{-}}L}\xspace}
\newcommand{\MSOiiL}{\ensuremath{\mathrm{MSO_2\text{-}}\Gamma}\xspace}
\newcommand{\PH}{\ensuremath{\mathrm{PH}}\xspace}

\newcommand{\threecolor}{\textsc{3-Colouring}}

\newcommand{\tw}{\mathop\mathit{tw}}

\def\prebox#1{\mathop{\textsl{#1}}\nolimits}
\def\thmso{\mbox{\rm Th}_{\MS}}
\def\Isups{I^{\mbox{\tiny$\subseteq$}}}
\newcommand{\adj}{\prebox{adj}}
\newcommand{\arc}{\prebox{arc}}

\def\roin{\mathrel{\in \textit{r.o.}}}
\def\rosubseteq{\mathrel{\subseteq \text{r.o.}}}
\def\SIZE{{\rm SIZE}}
\def\PH{{\rm PH}}
\def\SigmaCol#1{\ensuremath{\Sigma_{#1}}{\textsc{3Col}}}
\def\SigmaSat#1{\ensuremath{\Sigma_{#1}}\textsc{SAT}}
\def\Precol{{\it Precol}}

\usepackage{amsthm}
\newtheorem{theorem}{Theorem}[section]
\newtheorem{lemma}[theorem]{Lemma}
\newtheorem{corollary}[theorem]{Corollary}
\newtheorem{proposition}[theorem]{Proposition}

\newtheorem{remark}[theorem]{Remark}

\newtheorem{definition}[theorem]{Definition}
\newtheorem{example}[theorem]{Example}

\begin{document}

% \baselineskip1.5\baselineskip

\title{\bfseries Lower Bounds on the Complexity of \MSO Model-Checking
	\stepcounter{fnote}\footnote{A preliminary short version of this paper appeared in the Proceedings of STACS'12.}}
\author[fi]{Robert Ganian}
\ead{xganian1@fi.muni.cz}

\author[fi]{Petr Hlin\v en\'y}
% \stepcounter{fnote}\footnote{\underline{Corresponding
% 	author:} Petr Hlin\v en\'y, phone +420 721320618, fax +420 549491820,
% 	address Faculty of Informatics, Masaryk University,
% 	Botanick\'a 68a, 60200 Brno, Czech Republic.}}
\ead{hlineny@fi.muni.cz}

\author[rwth]{Alexander Langer}
\ead{langer@cs.rwth-aachen.de}

\author[fi]{Jan Obdr\v z\'alek}
\ead{obdrzalek@fi.muni.cz}

\author[rwth]{Peter Rossmanith}
\ead{rossmani@cs.rwth-aachen.de}

\author[rwth]{Somnath Sikdar}
\ead{sikdar@cs.rwth-aachen.de}

\address[rwth]{Computer Science, RWTH Aachen University, 
Germany\fnmark[fn1]}
\fntext[fn1]{Supported by Deutsche Forschungsgemeinschaft, project RO~927/9.}

\address[fi]{Faculty of Informatics, Masaryk University,
% Botanick\'a 68a, 60200 Brno\\
Czech Republic\fnmark[fn2]}
\fntext[fn2]{Supported by the Czech Science Foundation, project P202/11/0196.}

% \title{Lower Bounds on the Complexity of \MSO Model-Checking}
% \author{R. Ganian\inst{1}, P. Hlin\v{e}n\'y\inst{1}, A. Langer\inst{2},
%   J. Obdr\v{z}\'alek\inst{1}, P. Rossmanith\inst{2}, S. Sikdar\inst{2}}
%  \institute{\normalsize
%   Faculty of Informatics, Masaryk University, Brno, Czech Republic
%   \\\email{\{xganian1,hlineny,obdrzalek\}@fi.muni.cz} 
% 	\smallskip\and Theoretical
%   Computer Science, RWTH Aachen University, Germany
%   \\\email{\{langer,rossmani,sikdar\}@cs.rwth-aachen.de} }

\begin{abstract}
One of the most important algorithmic meta-theorems is a famous result by Courcelle,
which states that any graph problem definable in monadic second-order logic 
with edge-set quantifications (i.e., \MSOii model-checking) 
is decidable in linear time on any class of graphs of bounded tree-width. 
% In the parlance of parameterized complexity, this means that \MSOii\
% model-checking is 
% fixed-parameter tractable with respect to the tree-width as parameter. 
Recently, Kreutzer and Tazari~\cite{KT10b}
proved a corresponding complexity lower-bound---that 
\MSOii model-checking is not even in XP wrt.\ the formula size as parameter
for graph classes that are subgraph-closed and whose tree-width 
is poly-logarithmically unbounded. Of course, this is not an unconditional 
result but holds modulo a certain complexity-theoretic 
assumption, namely, the Exponential Time Hypothesis (ETH). 

% \looseness=-1%
In this paper we present a closely related result. We show that 
even \MSO model-checking with a fixed set of vertex labels, 
but without edge-set quantifications, is not in XP wrt.\ the formula 
size as parameter for graph classes which are subgraph-closed and 
whose tree-width is poly-logarithmically unbounded
unless the non-uniform ETH fails.
In comparison to Kreutzer and Tazari;
$(1)$~we use a stronger prerequisite, namely non-uniform instead of uniform
ETH, to avoid the effectiveness assumption and the construction of certain
obstructions used in their proofs; and $(2)$~we assume a different
set of problems to be efficiently decidable, namely
\MSO-definable properties on vertex labeled graphs instead of
\MSOii-definable properties on unlabeled graphs.

Our result has an interesting consequence in the realm of digraph width
measures: Strengthening the recent result~\cite{Ganianetal10}, we show that no
subdigraph-monotone measure can be ``algorithmically useful'', unless it is
within a poly-logarithmic factor of undirected tree-width.

\end{abstract}

\maketitle

\section{Introduction}
% We postpone basic formal definitions till Section~\ref{sec:preliminaries}.

A famous result by Courcelle, published in 1990, states that any graph
property definable in monadic second-order logic with quantification over
vertex- and edge-sets (\MSOii) can be decided in linear time on
any class of graphs of bounded tree-width \cite{Cou90}. 
More precisely, the \MSOii\ {\em model-checking}
problem for a graph $G$ of tree-width ${\rm tw}(G)$ and a formula $\varphi$, i.e.
the question whether $G\models\varphi$, can be solved in time 
$\ca O\big(|G|\cdot f(\varphi,{\rm tw}(G))\big)$. 
In the parlance of parameterized complexity, this means that \MSOii
model-checking is {\em fixed-parameter tractable (FPT)} 
with respect to the tree-width as parameter. 
% (In the world of parameterized complexity we say that such
% problems, solvable in time $\ca O(n^p\cdot f(k))$ for some constant $p$ and a
% computable function $f$, where $k$ is some parameter of the input and $n$
% the size of the input, are \emph{fixed-parameter tractable (FPT)}.) 

This result has a strong
significance. As \MSOii logic can express many interesting graph properties,
we immediately get linear-time algorithms for important \NP-hard problems,
such as \textsc{Hamiltonian Cycle}, \textsc{Vertex Cover}, and \textsc{3-Colorability}, 
on graphs of bounded tree-width. Such a result is called an \emph{algorithmic
  meta-theorem}, and many other algorithmic meta-theorems have since
appeared for other classes of graphs---see e.g.~\cite{Gro08,Kre11} for a
good survey.

As can be seen, Courcelle's theorem is a fast and relatively easy way of
establishing that a problem can be solved efficiently on graphs of bounded
tree-width. However, one may ask how far this result could be
generalized.  That is, is there another reasonable graph class 
of unbounded tree-width such that \MSOii
model-checking remains tractable on this class? Considering how important
this question is for theoretical understanding of what makes some problems 
on certain graph classes hard, it is surprising that until recently there has
not been much research in this direction.

For simplicity,
we postpone basic formal definitions to Sections~\ref{sec:preliminaries} and~\ref{sec:ktc}.

\subsection{Related prior work}

The first ``lower bound'' to Courcelle's theorem, by Makowski and
Mari\~{n}o, appeared in~\cite{mm03}. In that paper the authors show that if
a class of graphs has unbounded tree-width and is closed under topological
minors, then model-checking for \MSOii is not fixed-parameter tractable unless
$\PP = \NP$.
More recently, a stronger lower bound result by Kreutzer---not requiring the class
to be closed under minors---%
% the first result by Kreutzer providing a lower bound to Courcelle's theorem 
appeared in~\cite{Kre09CSL}. In that paper, Kreutzer
used the following version of ``unbounding'' the tree-width of a graph
class:

\begin{definition}[Kreutzer and Tazari \cite{Kre09CSL,KT10b}]
\label{def:strongly_unbounded}
The tree-width of a class $\cC$ of graphs is {\em strongly unbounded} by a
function $f\colon\N\to\N$ if there is $\epsilon<1$ and a polynomial $p(x)$
s.t. for all $n\in\N$ there is a graph $G_n\in\cC$ with the following properties:
\begin{enumerate}[i)]
\item the tree-width of $G_n$ is between $n$ and $p(n)$ and is greater than 
$f(|G_n|)$, and
\item given $n$, the graph~$G_n$ can be constructed in time $2^{n^\epsilon}$.
\end{enumerate}
The degree of the polynomial $p$ is called the \emph{gap-degree} of $\cC$
(with respect to $f$). The tree-width of $\cC$ is \emph{strongly unbounded
poly-logarithmically} if it is strongly unbounded by $\log^c{n}$, for all
$c\geq 1$.
\end{definition}

In other words, saying that tree-width of $\cC$ is \emph{strongly unbounded} means that 
\begin{enumerate}[(i)]
\item
there are no big gaps between the tree-width of witness graphs (those certifying 
that the tree-width of $n$-vertex graphs in $\cC$ is greater than $f(n)$), and 
\item
we can compute such witnesses effectively---in sub-exponential time wrt.~$n$.
\end{enumerate}

The main result of~\cite{Kre09CSL} is the following theorem:
\begin{theorem}[Kreutzer~\cite{Kre09CSL}]
Let $\Gamma$ be a fixed set of (at least two) colours, 
and $\cC$ be a class of graphs such that 
\begin{enumerate}[$(1)$]%\parskip-3pt
\item
 the tree-width of $\cC$ is strongly unbounded poly-logarithmically; 
\item
$\cC$ is closed under 
$\Gamma$-colourings (i.e., if $G\in\cC$ and $G'$ is obtained from $G$ 
by colouring some vertices or edges by colours from $\Gamma$, then $G'\in\cC$);
and, 
\item
$\cC$ is constructable (i.e., given a witness graph in $\cC$, a certain
substructure can be computed in polynomial time). 
\end{enumerate}
Then $\MC{\MSOiiL}$, the \MSOii model-checking problem on
all $\Gamma$-coloured graphs from~$\cC$,
is not in \XP (and hence not in \FPT---see Section~\ref{def:param} for a definition of these complexity
classes), unless all problems in the polynomial-time hierarchy can be solved
in sub-exponential time. 
\end{theorem}
This would, of course, mean that the Exponential-Time Hypothesis (ETH)~\cite{IPZ01} fails. 
The results of~\cite{Kre09CSL} have been improved
by Kreutzer and Tazari in~\cite{KT10a}, 
where the constructability requirement (3) was dropped.

A further improvement by the same authors appeared in~\cite{KT10b}.  
The main result in~\cite{KT10b} can be stated as follows: 
\begin{theorem}[Kreutzer and Tazari~\cite{KT10b}]
Let $\cC$ be a class of graphs such that 
\begin{enumerate}%[$(2')$]%\parskip-3pt
\item[$(1)$]
the tree-width of $\cC$ is strongly unbounded poly-logarithmically; and 
\item[$(2')$]
$\cC$ is closed under taking subgraphs, i.e. $G\in\cC$ and $H\subseteq G$ implies $H\in\cC$. 
\end{enumerate}
Then \MC{\MSOii}, the \MSOii model-checking problem on~$\cC$, is not in \XP unless
all problems in the polynomial-time hierarchy can be solved in
sub-exponential time.
\end{theorem}
Note that $(2')$, to be closed under subgraphs, is a strictly weaker condition
than previous $(2)$, to be closed under $\Gamma$-colourings (of edges, too).

\subsection{New contribution}
In this paper we prove a result closely related to 
Kreutzer--Tazari's~\cite{Kre09CSL,KT10a,KT10b} but for 
\MSO\ logic with a fixed set of vertex 
labels.
The role of {\em vertex labels} in our paper is similar to that of colours 
in~\cite{Kre09CSL,KT10a}, but weaker in the sense that the labels are not assigned to edges.%
\footnote{The reason we use the term labels and not colours is to
be able to clearly distinguish between vertex-labeled graphs and the
coloured graphs used in~\cite{Kre09CSL,KT10a}, where colours are
assigned to edges and vertices.
}
In contrast to the work by Kreutzer and Tazari, 
we assume a different set of problems---those expressible by \MSOL
on graphs with vertex labels from a fixed finite set $L$\,---to be efficiently solvable
on a graph class in order to derive an analogous conclusion.

Before stepping further, we mention one more fact. There exist
classes $\cC$ of $L$-labeled graphs of unbounded tree-width on which \MC{\MSOL},
the \MSO model-checking problem on $\cC$, is
polynomial time solvable, e.g.\ classes of bounded clique-width or
rank-width. But it is important to realize that these classes are 
{\em not} closed under taking subgraphs. 

Our main result then reads---cf.~Section~\ref{sec:tmt}:

\begin{theorem}[reformulated as Theorem~\ref{thm:main}] \label{thm:mainI}
Assume a (suitable but fixed) finite label set $L$, and a graph class $\cG$ 
satisfying the following two properties:
\begin{enumerate}[a)]
\item $\cG$ is closed under taking subgraphs and under $L$-vertex-labelings,
\item the tree-width of $\cG$ is densely unbounded poly-logarithmically
(see Def.~\ref{defn:unbounded_polylog}).
\end{enumerate}
Then \MCa{\MSOL}{\cG}, the \MSOL model-checking problem on all $L$-vertex-labeled
graphs from $\cG$, is {\em not} in \XP unless the non-uniform 
Exponential-Time Hypothesis fails.
\end{theorem}

Our general approach follows that by Kreutzer and Tazari in~\cite{Kre09CSL,KT10a,KT10b}
but differs from theirs in three main ways:
\begin{enumerate}[(I)]
\item
Kreutzer and Tazari require witnesses as in (ii) of 
Definition~\ref{def:strongly_unbounded} of~\cite{KT10b} 
to be computable effectively in their proofs. 
It is unclear how this can be done,
and hence they simply add this as a natural requirement on~$\cC$. 
%This seems to be a highly nontrivial task, and so the authors have decided
%to simply add this as a reasonable and allegedly natural requirement in
%their definition.
Furthermore, the construction of certain obstructions (grid-like minors) 
used in their proof requires an involved machinery~\cite{KT10a}.
We adopt a different position (note our ``densely unbounded''
in Definition~\ref{def:denselyunbounded} vs.\
``strongly unbounded'') and avoid both aspects by using a stronger 
complexity-theoretic assumption, namely the non-uniform ETH instead of 
the ordinary ETH. In this way, we 
can get the obstructions as \emph{advice} ``for free.''

\item
Our result applies to \MSOL model-checking on {\em$L$-vertex-labeled graphs},
while the result of~\cite{KT10b} applies to \MSOii over unlabeled graphs.
There are problems that can be expressed in \MSOL and
not in \MSOii and vice versa (take \textsc{Red-Blue Dominating Set}
vs.\ \textsc{Hamiltonian Cycle}, for instance).
If, however, the set of labels~$L$ is fixed for both,
\MSOL has much weaker expressive power than \MSOii-$L$ 
due to missing edge-set quantifications (see Section~\ref{sec:preliminaries}). 
In particular, note that many of the existing algorithmic
meta-theorems (e.g.~\cite{Cou90,CMR00})
that deal with \MS-definable properties handle unlabeled as well as
(vertex-)labeled inputs with equal ease. However, extending e.g.\ the results 
of~\cite{CMR00} from \MSOL to \MSOii is not possible unless 
$\EXP = \NEXP$.

\item
Finally, because of the free advice, our proof does not need technically
involved machinery such as the simulation of a run 
of a Turing machine encoded in graphs~\cite{KT10b}. 
This makes our proof shorter and exhibits its structure more clearly.
\end{enumerate}
 
After all, Theorem~\ref{thm:mainI} gives a good indication (II) that 
poly-logarithmically unbounded tree-width along
with closure under subgraphs is a strong enough condition for even the 
{\em bare \MSO model-checking to be intractable} (modulo appropriate
complexity-theoretic assumptions).  

Moreover, if we assume that the label set~$L$ is potentially unbounded, 
then we obtain a stronger result (getting us even ``closer'' to \cite{KT10b}):
\begin{theorem}[reformulated as~Theorem~\ref{thm:PH_subexp}]
\MSOL model-checking with vertex labels~$L$ ($L$ depending on the formula size)
is not tractable for a graph class satisfying (a)~and~(b) 
of Theorem~\ref{thm:mainI} unless \emph{every} problem 
in the polynomial-time hierarchy is in $\subexpsubexp{n}$.
\end{theorem}

Finally, as a corollary, we obtain an interesting consequence in the area
of directed graph (digraph) width measures, improving upon~\cite{Ganianetal10}.
% In this context, we let~$U(D)$ denote the underlying undirected
% graph of a digraph~$D$. Given a digraph width measure~$\delta$,
% we let~$U_{\delta}(d) := \{U(D) \mid \delta(D) \le d\}$ to be the set
% of underlying undirected graphs of digraphs of $\delta$-width at most~$d$.  

\begin{theorem}[reformulated as Theorem~\ref{thm:maindir}]
\label{thm:maindirI}
Assume a (suitable but fixed) finite label set $L$,
and a digraph width measure $\delta$ such that
\begin{enumerate}[a)]
\item
$\delta$ is monotone under taking subdigraphs and $L$-vertex-labelings, and
\item
\MCa{\MSOL}{\cD},
the \MSOL model-checking problem on all $L$-vertex-labeled digraphs from $\cD$,
is in \XP wrt.\ $\delta(D)$ and $\varphi$ as parameters.
\end{enumerate}
Then, unless the non-uniform ETH fails, for all~$d \in \N$ the
tree-width of the class~$U_{\,\delta}(d) := \{U(D) \mid \delta(D) \le d\,\}$,
the underlying undirected graphs of digraphs of $\delta$-width at
most~$d$, is not densely unbounded poly-logarithmically.
\end{theorem}
Informally, a digraph width measure that is subdigraph-monotone
and algorithmically ``powerful'' is at most a poly-logarithmic factor 
of the tree-width of the underlying undirected graph---cf.~Section~\ref{sec:directed}.

\paragraph*{Paper organization}
% The rest of the paper is organized as follows: 
In Section \ref{sec:preliminaries} we overview some standard terminology and notation.
Section \ref{sec:ktc} then includes the proof outline
and the core technical concepts:
unbounding tree-width (Definition~\ref{def:denselyunboundedII}), 
the grid-like graphs of Reed and Wood~\cite{RW08} (Proposition~\ref{thm:ReedWood}),
and a new way of interpreting arbitrary graphs in labeled
grid-like graphs of sufficiently high order (Lemma~\ref{lem:keyinterpret}).
These then lead to the proof of our main result, equivalently formulated as
Theorem~\ref{thm:main}, in Section~\ref{sec:tmt}.
Two extensions of the main result appear in Section~\ref{sec:tmt-s};
the first one discussing a stronger collapse of PH (under allowing
non-fixed labeling of graphs), and the second one
considering classes of (just) poly-logarithmically unbounded tree-width,
i.e.\ those which may not be strongly/densely unbounded.
A consequence for directed width measures is then discussed in
Section~\ref{sec:idwm}, followed by concluding remarks in
Section~\ref{sec:concluding}.

\section{Preliminaries}
\label{sec:preliminaries}
\subsection{Graphs}
The graphs we consider in this paper are \emph{simple}, i.e.\ they do not
contain loops and parallel edges.  Given a graph~$G$, we let~$V(G)$ denote
its vertex set and~$E(G)$ its edge set. A \emph{path $P$ of length}
$r>0$ in $G$ is a sequence of vertices $P=(x_0,\ldots, x_r)$ such that all
$x_i$ are pairwise distinct and $(x_i, x_{i+1})\in E(G)$ for every $0\le i<
r$. Let $\cS$ be a family of sets $S_i$ for $i=1,2,\ldots$. Then
the \emph{intersection graph on} $\cS$ is the graph $I(\cS)$ where
$V(I(\cS))=\cS$ and $S_iS_j\in E(I(\cS))$ iff $S_i\cap S_j\not=\emptyset$.

Let $L=\{L_1,\ldots,L_k\}$ be a set of labels. A \emph{$L$-vertex-labeled graph},
or \emph{$L$-graph} for short, is a graph $G$ together with a function
$\lambda \colon V(G) \to 2^L$, assigning each vertex a set of labels, and we write
$(G,\lambda)$ to denote this graph. 
For a graph class $\cG$, we shortly write $\cG^L$ for the class of all
$L$-graphs over $\cG$, i.e.\ $\cG^L$ contains all $(G,\lambda)$
where $G \in \cG$ and $\lambda$ is an arbitrary $L$-vertex-labeling of~$G$.
Note that, unlike in e.g.~\cite{Kre09CSL}, we do not allow labels for edges,
which is in accordance with our focus on \MSO logic of graphs.

\subsection{\MS logic on graphs} 
Monadic second-order logic ($\MS$) is an extension of first-order
logic by quantification over sets. 
On the one-sorted adjacency model of graphs it reads as follows:

\begin{definition}
\label{df:MSO1}
The language of $\MSO$, \emph{monadic second-order logic of graphs},
contains the expressions built from the following elements:
\begin{enumerate}[i)]\parskip-3pt
\item variables $x,y,\ldots$ for vertices, and $X,Y,\ldots$ for sets of vertices,
\item the predicates $x\in X$ and $\adj(x, y)$ with the standard meaning,
\item
   equality for variables, 
   the connectives~$\land, \lor, \lnot, \to$ and the
   quantifiers~$\forall, \exists$\,.
\end{enumerate}
\end{definition}
\smallskip
Note that we do not allow quantification over sets of edges
(as edges are not elements). If we considered the two-sorted incidence graph
model (in which the edges formed another sort of elements),
we would obtain aforementioned $\MSOii$, \emph{monadic second-order logic of graphs with 
edge-set quantification}, which is strictly more powerful
than~$\MSO$, cf.~\cite{EF99}.  Yet even $\MSO$ has
strong enough expressive power to describe many common problems. 

\begin{example}\label{ex:3colorMso1}
The \threecolor\ problem can be expressed in $\MSO$ as follows:
\begin{align*}
 \exists V_1,V_2,V_3 \allowbreak \big[\, &
 	\forall v \> (v\in V_1\vee v\in V_2\vee v\in V_3)
 \,\wedge \\ &
 \bigwedge\nolimits_{i=1,2,3}\>
   \forall v,w \>
 	(v\not\in V_i\vee w\not\in V_i\vee \neg\adj(v,w))
 \>\big]
\end{align*}
\end{example}

The $\MSO$ logic can naturally be extended to $L$-graphs. The
\emph{monadic second-order logic} on $L$-vertex-labeled graphs, denoted by
$\MSOL$, is the natural extension of $\MSO$ with unary predicates $L_i(x)$ for
each label $L_i\in L$, such that $L_i(x)$ holds iff $L_i\in\lambda(x)$.

\subsection{Parameterized complexity and $\MSO$ model-checking}
\label{def:param}
Throughout the paper we are interested in the problem of checking
whether a given input graph satisfies a property specified by a fixed
formula $\varphi$. This problem can be thought of as an  instance of a problem
parameterized by $\varphi$, as studied in the field of  
\emph{parameterized complexity} (see e.g.~\cite{fg06} for a background on parameterized complexity).

A parameterized problem~$Q$ is a subset of~$\Sigma \times \mathbb{N}_0$, where~$\Sigma$
is a finite alphabet and~$\mathbb{N}_0 = \mathbb{N} \cup \{0\}$. 
A parameterized problem~$Q$ is said to be \emph{fixed-parameter tractable}
if there is an algorithm that given~$(x,k) \in \Sigma \times \mathbb{N}_0$ decides whether~$(x,k)$
is a yes-instance of~$Q$ in time~$f(k) \cdot p(|x|)$ where~$f$ is some
computable function of~$k$ alone, $p$ is a polynomial and~$|x|$ is the size measure of
the input. 
The class of such problems is denoted by \FPT.
The class \XP is the class of parameterized problems that admit 
algorithms with a run-time of~$\ca O(|x|^{f(k)})$ for some computable~$f$,
i.e.\ polynomial-time for every fixed value of~$k$.

We are dealing with a parameterized {\em model-checking problem}
\MC{\MSO} where $\cC$ is a class of graphs;
the task is to decide, given a graph $G\in\cC$ and a formula
$\varphi\in\MSO$, whether $G\models\varphi$. 
The parameter is $k=|\varphi|$, the size of the formula $\varphi$.
We actually consider the labeled variant \MC{\MSOL} for $\cC$ being a class of
$L$-graphs.

\subsection{Interpretability of logic theories}
One of our main tools is the classical {interpretability} of logic theories~\cite{Rab64}
(which in this setting is analogical to transductions as used e.g.\ by
Courcelle, cf.~\cite{CE12}).
To describe the simplified setting, assume 
that two classes of {\em relational structures} $\mathscr K$ and $\mathscr L$ are given.
The basic idea of an {\em interpretation} $I$ of the {theory} $\thmso(\mathscr K)$
into $\thmso(\mathscr L)$ is to transform \MS formulas $\varphi$ over $\mathscr K$
into \MS formulas $\varphi^I$ over $\mathscr L$
in such a way that ``truth is preserved'':
\begin{itemize}
\item
First, one chooses a formula $\alpha(x)$ intended to define in each structure
$G\in\mathscr L$ a set of individuals (new domain)
$G[\alpha]:=\{a:a\in dom(G)\mbox{ and }G\models\alpha(a)\}$, 
where $dom(G)$ denotes the set of individuals (domain) of $G$.
\item
Then, one chooses for each $s$-ary relational symbol $R$ from $\mathscr K$ 
a formula $\beta^R(x_1,\dots,x_s)$, 
with the intention to define a corresponding relation 
$G[\beta^R]:=\{(a_1,\dots,a_s) \colon\allowbreak a_1,\dots,a_s\in dom(G)$ and 
$G \models\beta^R(a_1,\dots,a_s)\}$. 
With these formulas one defines for each $G\in\mathscr L$ the relational
structure $G^I := \big(G[\alpha], G[\beta^R], \dots\big)$
intended to correspond with structures in $\mathscr K$.
\item
Finally, there is a natural way to
translate each formula $\varphi$ (over $\mathscr K$) into a formula $\varphi^I$
(over $\mathscr L$),
by induction on the structure of formulas. 
The atomic ones are substituted by corresponding chosen formulas 
(such as $\beta^R$) with the corresponding variables. 
Then one proceeds via induction simply as follows:
\begin{eqnarray*}
(\neg \varphi)^I ~\mapsto~ \neg(\varphi^I) &, &\quad
 (\varphi_1 \wedge \varphi_2)^I ~\mapsto~ (\varphi_1)^I \wedge (\varphi_2)^I, 
\\
\left(\exists x^{\vbox to1ex{\vfill}}\, \varphi(x)\right)^I ~\mapsto~ 
        \exists y \left(\alpha(y) \wedge \varphi^I(y)\right) &,&\quad
\left(\exists X^{\vbox to1ex{\vfill}}\, \varphi(X)\right)^I ~\mapsto~ 
        \exists Y\, \varphi^I(Y).
\end{eqnarray*}
\end{itemize}
The whole concept is shortly illustrated in by the scheme
in Figure~\ref{fig:interpret}.
\begin{figure}[tb]
$$
      \begin{matrix}
        \varphi\in \mbox{\MS over $\mathscr K$}\cr H\in\mathscr K
% 	\cr~\cr G^I\cong H
	\cr~\cr G^I\cong H\cr (\mbox{s.t. }G^I\models\varphi)
      \end{matrix}
\quad
\begin{matrix}
  I\cr \raise1ex\hbox{ $-\!\!\!-\!\!\!-\!\!\!-\!\!\!\longrightarrow$ }
	\medskip
	\cr \lower1ex\hbox{$I$}\cr \hbox{$\longleftarrow\!\!\!-\!\!\!-\!\!\!-\!\!\!-$ }
\end{matrix}
\quad
\begin{matrix}
  \varphi^I\in \mbox{\MS over $\mathscr L$}\cr G\in\mathscr L
% 	\cr~\cr G
	\cr~\cr G\cr (\mbox{s.t. }G\models\varphi^I)
\end{matrix}
$$
\caption{A basic scheme of an interpretation of $\thmso(\mathscr K)$
into $\thmso(\mathscr L)$.}
\label{fig:interpret}
\end{figure}

\begin{definition}[Interpretation between theories]\label{df:interpret}
Let $\mathscr K$ and $\mathscr L$ be classes of relational structures.
Theory $\thmso(\mathscr K)$ is {\em interpretable} in theory $\thmso(\mathscr L)$
if there exists an interpretation $I$ as above such that the following two
conditions are satisfied:
\begin{enumerate}[i)]\parskip-3pt
\item
For every structure $H\in\mathscr K$, there is  $G\in\mathscr L$ 
        such that $G^I\cong H$, and
\item
for every $G\in\mathscr L$, the structure $G^I$ is isomorphic to 
        some structure of $\mathscr K$.
\end{enumerate}
Furthermore, $\thmso(\mathscr K)$ is {\em efficiently interpretable} in
$\thmso(\mathscr L)$ if %(cf.\ Figure~\ref{fig:interpret})
the translation of each $\varphi$ into $\varphi^I$
is computable in polynomial time
and the structure $G\in\mathscr L$, where $G^I\cong H$,
can be computed from any $H\in\mathscr K$ in polynomial time.
\end{definition}

\subsection{Exponential-Time Hypothesis}

The \emph{Exponential-Time Hypothesis (ETH)}, formulated in~\cite{IPZ01},
states that there exists no algorithm that can solve $n$-variable
\textsc{3-SAT} in time $2^{o(n)}$. It was shown in~\cite{IPZ01} that the
hypothesis can be formulated using one of the many equivalent problems
(e.g. \textsc{$k$-Colourability} or \textsc{Vertex Cover})---i.e. sub-exponential complexity
for one of these problems would imply the same for all the others.

ETH can be formulated in the {\em non-uniform} version: There is no family of
algorithms (one for each input length) which can solve $n$-variable
\textsc{3-SAT} in time $2^{o(n)}$. In theory of computation literature,
``non-uniform algorithms'' are often referred to as ``fixed-sized input
circuits'' where for each length of the input a different circuit is
used. Yet another way of thinking about non-uniform algorithms is as having
an algorithm that is allowed to receive an oracle advice, which depends
only on the length of the input. As mentioned in~\cite{CSH08}, the results
of~\cite{IPZ01} hold also for the non-uniform~ETH.

\section{Key Technical Concepts}
\label{sec:ktc}

\paragraph*{Proof outline}
We are going to show via a suitable multi-step reduction, that
the potential tractability of \MSOL model-checking on our graph class
$\cG$ (whose tree-width is densely unbounded poly-logarithmically),
implies sub-exponential time algorithms for problems which are not
believed to have one (cf.~ETH).
The success of the reduction, of course, rests on the assumptions of $\cG$
being subgraph-closed and of unbounded tree-width.
So, at a high level, our proof technique is similar to that of Kreutzer and Tazari. 

However, there are some crucial differences.
While \cite{KT10b} uses the effectiveness assumption in
Definition~\ref{def:strongly_unbounded}.\,ii and some further technically
involved algorithms to construct a ``skeleton'' in the class $\cC$
suitable for their reduction,
in our reduction we will obtain a corresponding labeled skeleton in the class
$\cG^L$ ``for free'' from an oracle advice function
which comes with the non-uniform computing model.
That is why our complete proof is also significantly shorter than that in~\cite{KT10b}.
Additionally, our arguments shall employ a result on strong edge colourings of
graphs in order to ``simulate'' certain edge sets within the \MSOL language,
thus avoiding the need for a more expressive logic such as \MSOii.

\subsection{Unbounding Tree-width}
Following Definition~\ref{def:strongly_unbounded}, we aim to formally
describe what it means to say that the tree-width of a graph class is not bounded
by a function~$g$.
Recall (see also \cite{Kre09CSL,KT10b})
that it is not enough just to assume $\tw(G)>g(|V(G)|)$ for some
sporadic values of $\tw$ with huge gaps between them, but a reasonable
density of the surpassing tree-width values is also required.
Hence we suggest the following definition as a weaker alternative
to Definition~\ref{def:strongly_unbounded}:

\begin{definition}[Densely unbounded tree-width]
\label{def:denselyunbounded}
For a graph class $\cG$, we say that the tree-width of $\cG$ is
{\em densely unbounded by a function $g$} if there is a constant~$\gamma>1$
such that, for every $m\in\N$, there exists a graph $G\in\cG$
whose tree-width is $\tw(G)\geq m$ and $|V(G)|<\cO\big(g^{-1}(m^\gamma)\big)$.
The constant $\gamma$ is called the {\em gap-degree} of this property.
\end{definition}

\begin{remark}
Comparing to Definition~\ref{def:strongly_unbounded}
one can easily check that if the tree-width of a class $\cG$ is strongly
unbounded by a function $g$, then the tree-width is densely unbounded by $g$
with the same gap-degree,
and the witnessing graphs $G$ of Definition~\ref{def:denselyunbounded}
can be computed for all $m$ efficiently---in sub-exponential time wrt.~$m$.
Hence our definition is weaker in this respect.
\end{remark}

For simplicity, we are interested in graph classes whose tree-width is densely
unbounded by every poly-logarithmic function of the graph size.
That is expressed by the following simpler definition:

\begin{definition}[Densely unbounded tree-width II] \label{defn:unbounded_polylog}
\label{def:denselyunboundedII}
For a graph class $\cG$, we say that the tree-width of $\cG$ is
{\em densely unbounded poly-logarithmically} if it is densely unbounded by
$\log^c m$ for every $c\in\N$. 

That is, for every~$c >1$ the following holds:
for all $m \in \N$ there exists a graph $G \in \cG$
whose tree-width is $\tw(G) \geq m$ and with size $|V(G)|<\cO\big(2^{m^{1/c}}\big)$.
(The gap-degree becomes irrelevant in this setting.)
\end{definition}

\subsection{Grid-like graphs}
The notion of a grid-like minor has been introduced by Reed and Wood in~\cite{RW08},
and extensively used by Kreutzer and Tazari~\cite{KT10a,KT10b}.
In what follows, we avoid use of the word ``minor'' in our definition of the same
concept, since ``$H$-minors'' where~$H$ is grid-like are always found
as subgraphs of the target graph, which might cause some confusion.

\begin{definition}[Grid-like \cite{RW08}]
\label{def:grid-like}
A graph $G$ together with a collection $\cP$ of paths,
formally the pair $(G,\cP)$, is called {\em grid-like}
if the following is true:
\begin{enumerate}[i)]\parskip-3pt
\item
$G$ is the union of all the paths in $\cP$,
\item\label{it:path2v}
each path in $\cP$ has at least two vertices, and
\item
the {\em intersection graph} $I(\cP)$ of the path collection is bipartite.
\end{enumerate}
The {\em order} of such grid-like graph $(G,\cP)$ is the maximum integer $\ell$ such
that the intersection graph~$I(\cP)$ contains a $K_\ell$-minor.
When convenient, we refer to a grid-like graph simply as to~$G$.
\end{definition}
Note that the condition (\ref{it:path2v}) is not explicitly stated 
in~\cite{RW08}, but its validity implicitly follows from the point to get a
$K_\ell$-minor in $I(\cP)$, cf.~Theorem~\ref{thm:ReedWood}.
Since the traditional square (and hexagonal, too) grids are grid-like
with the horizontal and vertical paths forming the collection $\cP$,
the new concept of having a grid-like subgraph generalizes the traditional
concept of having a grid-minor.
See also Figure~\ref{fig:grid}.

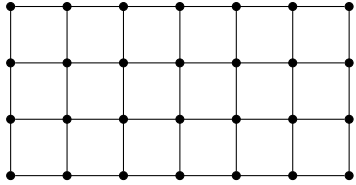
\begin{figure}[tb]
  \centering
\begin{tikzpicture}[scale=0.75]
\tikzstyle{every node}=[draw, shape=circle, minimum size=3pt,inner sep=0pt, fill=black]

  \foreach \i in {0,...,6}
    \foreach \j in {0,...,3}
      \node at (\i,\j) (x\i\j) {};
  \foreach \i in {0,...,6}
    \draw (x\i0)--(x\i3);
  \foreach \j in {0,...,3}
    \draw (x0\j)--(x6\j);

\end{tikzpicture}
  \caption{The square $4\times7$-grid is grid-like of order $5$
	(where $\cP$ is the collection of the horizontal and the vertical paths
	and $I(\cP)\simeq K_{4,7}$).}
  \label{fig:grid}
\end{figure}

One can easily observe the following:

\begin{proposition}
\label{prop:grid-like}
Let $(G,\cP)$ be a grid-like graph.
Then the collection $\cP$ can be split into $\cP=\cP_1\cup\cP_2$
such that each $\cP_i$, $i=1,2$, consists of pairwise disjoint paths.
Consequently, the maximum degree in $G$ is $\Delta(G)\leq4$.
\end{proposition}

The next result is crucial for our paper 
(while we do not require constructability as in~\cite{KT10b}):
\begin{theorem}[Reed and Wood \cite{RW08}]
\label{thm:ReedWood}
Every graph with tree-width at least $c\ell^4\sqrt{\log\ell}$ contains 
a subgraph which is grid-like of order $\ell$, for some constant $c$.
\end{theorem}

\subsection{\MSO interpretation on grid-like graphs}

Now we prove the core new technical tool of our paper;
showing how the subgraphs of $I(\cP)$ of any grid-like graph $(G,\cP)$ can be
efficiently $\MSO$-interpreted in $G$ itself with a suitable vertex labelling.
First, we state a useful result about strong edge colourings of graphs---a
{\em strong edge-colouring} is an assignment of colours to the edges of a graph 
such that no path of length three contains the same colour twice.
\begin{theorem}[Cranston~\cite{Cra06}]
\label{thm:strongedgecolouring}
Every graph of maximum degree $4$ has a strong edge-colouring using at most $22$ colours.
This colouring can be found with a polynomial-time algorithm.
\end{theorem}

For  a class of grid-like graphs~$\cG$,
let~$\Isups(\cG) = \{H: H\subseteq I(\cP),\> (G,\cP)\in\cG\}$ 
denote the class of all subgraphs of their intersection graphs.
Our core tool is the following lemma.
\begin{lemma}
\label{lem:keyinterpret}
Let $\cG$ be any class of grid-like graphs.
There exists a fixed finite set $L$ of labels, with~$|L|\geq 47$,
and a graph class $\cI \supseteq \Isups(\cG)$, % depending only on~$\cG$,
such that the following holds.
The $\MSO$ theory of $\cI$ has an efficient interpretation
in the $\MSO$ theory of \,$\cG^L$\,---the class of all $L$-vertex-labeled graphs over~$\cG$.
Stated differently, any $H \subseteq I(\cP)$ where $(G,\cP) \in \cG$ is interpreted in some 
$L$-graph of~$G$.
\end{lemma}

\begin{proof}
Note that the use of a class $\cI$ in the statement of the lemma is
only a technicality related to (ii) of Definition~\ref{df:interpret}.
We are actually interested only in interpreting the graphs from $\Isups(\cG)$,
and $\cI$ then simply contains all the graphs that (also accidentally)
result from the presented interpretation.

Hence we choose an arbitrary $(G,\cP)\in\cG$ and $H\subseteq I(\cP)$.
The task is to find a vertex labeling $\lambda_H\colon V(G)\to2^L$
such that $H$ has an efficient $\MSO$ interpretation 
in the labeled graph $(G,\lambda_H)\in \cG^L$.
By Theorem~\ref{thm:strongedgecolouring} (cf.\ also Proposition~\ref{prop:grid-like}),
let $\gamma\colon E(G) \to \{1, \ldots, 22\}$ be a strong edge-colouring of the
chosen graph $G$.
Let $\cP = \cP_w\cup\cP_b$ be the bipartition of the paths forming $G$
corresponding to the partite sets of~$I(\cP)$.
We call the paths of $\cP_w\cap V(H)$ ``white'' 
and those of $\cP_b\cap V(H)$ ``black''.
The remaining paths not in the vertex set of $H$ are irrelevant.
The edges of white/black paths are also called white/black, respectively,
with the understanding that some edges of $G$ may be both white and black.
For $x\in V(G)$,
we let $w(x)=\{\gamma(f): f\mbox{ is a white edge incident to $x$}\}$
and $b(x)=\{\gamma(f): f\mbox{ is a black edge incident to $x$}\}$.
According to Proposition~\ref{prop:grid-like},
$|w(x)|\leq2$ and $|b(x)|\leq2$.

The key observation, derived directly from the definition of a strong
edge-colouring, is that any edge $f=xy\in E(G)$ is a white edge
iff $w(x)\cap w(y)\not=\emptyset$, and analogously for black edges.
This allows us to speak separately about the white and black edges
in $G$ using only the language of $\MSOL$.
Another easy observation is that the vertex sets of the paths in $\cP$
have a system of distinct representatives by Hall's theorem. For if $\cP' \subseteq \cP$
and~$\cP'$ contains~$p$ white paths and~$q$ black paths, 
then $|V(\cP')| \geq 2 \cdot \max \{p,q\} \geq p+q$, proving Hall's criterion.
We assign a marker $r(x)\in\{\emptyset,w,b\}$ to each $x\in V(G)$ such that
$r^{-1}(w)$ is the set of the representatives of white paths
and $r^{-1}(b)$ is that of black paths
(i.e., $r^{-1}(\emptyset)$ are not representatives).
Finally, we assign another vertex marker $m(x)\in\{0,1\}$ to each vertex $x\in V(G)$
such that $m(x)=1$ iff $x \in V(P_1) \cap V(P_2)$ where 
$P_1,P_2\in V(H)\subseteq\cP$ and $\{P_1,P_2\}\in E(H)$.

Hence the label set $L$ consists of $22$ ``light'' colours coming from
$\gamma$ values on white paths, another $22$ ``dark'' colours 
from black paths, and the three singletons $w,b,m$ described above
(altogether $47$ binary labels).
Note that the actual size of the needed label space over $L$ is even much smaller;
at most~$\big[{22\choose2}+22+1\big]^2 \cdot 3 \cdot 2 < 2^{19}$.
The label $\lambda_H(x)$ of a vertex $x\in V(G)$ then contains the disjoint
union $w(x)\dot\cup\,b(x)$, the label $r(x)$ if $\not=\emptyset$,
and finally $m$ if $m(x)=1$.

Now, the interpretation of $H$ in $(G,\lambda_H)$ is simply as follows:
The domain, i.e.\ the vertex set of $H$, is identified within $V(G)$
by a predicate $\alpha(x)$ expressing that $``r(x)=w\vee r(x)=b"$ in $\MSOL$.
In formal logic language (cf.\ Section~\ref{sec:preliminaries}), 
it is $L_w(x)\vee L_b(x)$.
The relational symbol $\prebox{adj}$ of $H$ is then replaced,
for $x,y\in V(G)$ s.t.\ $\alpha(x)\wedge\alpha(y)$, with
\begin{align*}
\beta^{\prebox{adj}}(x,y) &\,\equiv\, \exists z
	\,\big[ ``m(z)=1" \wedge \varrho(x,z)\wedge \varrho(y,z) \big], \mbox{ where}
\\
\varrho(t,z) &\,\equiv\, \big[ ``r(t)=w" \to \prebox{con}_w(t,z) \big]
	\wedge \big[ ``r(t)=b" \to \prebox{con}_b(t,z) \big]
\end{align*}
and where $\prebox{con}_w$ ($\,\prebox{con}_b\,$) routinely expresses
in $\MSOL$ the fact that $t,z$ belong to the same component induced
by white (black) edges in~$G$.
Precisely,
\begin{align*}
\prebox{con}_w(t,z) \,\equiv\, \forall Z & \big[
	\, z\in Z\wedge t\not\in Z \,\to\, \exists\, u,v
\\ & \big( v\in Z\wedge u\not\in Z\wedge \prebox{adj}(u,v)\wedge
			``w(u)\cap w(v)\not=\emptyset" \big)\big]
.\end{align*}
Clearly, in this interpretation $(G,\lambda_H)^I\simeq H$ thanks to our
choice of $\lambda_H$.
This completes the proof.
\end{proof}

Lemma~\ref{lem:keyinterpret} will be coupled with the next technical
tool of similar flavor used in our previous~\cite{Ganianetal10}.
We remark that its original formulation was even stronger, making the target
graph class planar, but we are content with the following weaker formulation here.
We call a graph $G$ {\em$\{1,3\}$-regular} if all the vertices of $G$
have degree either one or three.

\begin{lemma}[{\cite[in Theorem 5.5]{Ganianetal10}}]
\label{lem:interpretin13regular}
The $\MSO$ theory of all simple graphs has an efficient interpretation
in the $\MSO$ theory of all simple $\{1,3\}$-regular graphs.
Furthermore, this efficient interpretation $I$ can be chosen such that,
for every $\MSO$ formula $\psi$, the resulting property $\psi^I$ is
invariant under subdivisions of edges; i.e.\ for every $\{1,3\}$-regular
graph $G$ and any subdivision $G_1$ of $G$ it holds
$G\models\psi^I$ iff $G_1\models\psi^I$.
\par\vskip-1ex\qed
\end{lemma}

\section{The Main Theorem}
\label{sec:tmt}

\begin{theorem}[cf.\ Theorem~\ref{thm:mainI}]
\label{thm:main}
Let $L$ be a finite set of labels, $|L|\geq47$.
Unless the nonuniform Ex\-po\-nen\-tial-Time Hypothesis fails, 
there exists no graph class $\cG$ satisfying all the three properties
\begin{enumerate}[a)]
\item\label{it:clsubgraph}
$\cG$ is closed under taking subgraphs,
\item\label{it:densunb}
the tree-width of $\cG$ is densely unbounded poly-logarithmically,
\item\label{it:MSOinXP}
the \MCa{\MSOL}{\cG^L} model-checking problem is in \XP,
i.e., one can test whether $G\models\varphi$ in time 
$\cO\big(|V(G)|^{f(|\varphi|)}\big)$ for some computable function~$f$.
\end{enumerate}
\end{theorem}

\begin{figure}
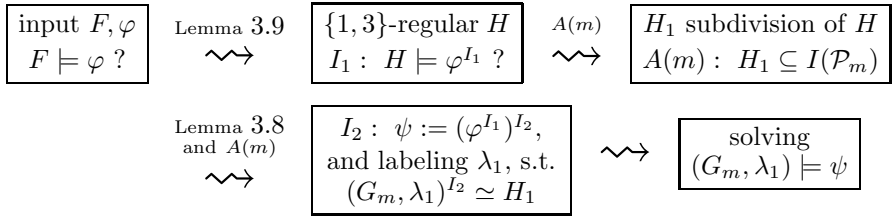

\begin{align*}
\boxed{{\mbox{input $F,\varphi$}\atop F\models\varphi ~?}}\quad&
	\substack{{\rm Lemma~\ref{lem:interpretin13regular}}\\[1ex]\mbox{\huge$\leadsto$}}
\quad\boxed{{\mbox{$\{1,3\}$-regular $H$}\atop I_1:~ H\models\varphi^{I_1} ~?}}\quad
	\substack{{A(m)}\\[1ex]\mbox{\huge$\leadsto$}}
\quad\boxed{{\mbox{$H_1$ subdivision of $H$}\atop A(m):~ H_1\subseteq I(\cP_m)}}\quad
\\[1ex]
	&\substack{{\rm Lemma~\ref{lem:keyinterpret}}\\[1pt]{{\rm and}~A(m)}
			\\[1ex]\mbox{\huge$\leadsto$}}
\quad\boxed{{ I_2:~ \psi:=(\varphi^{I_1})^{I_2}, \atop
		{ \mbox{and labeling $\lambda_1$, s.t.} \atop 
			\mbox{$(G_m,\lambda_1)^{I_2}\simeq H_1$} }}}\quad
	\substack{{}\\\mbox{\huge$\leadsto$}}
\quad\boxed{\mbox{solving}\atop \mbox{$(G_m,\lambda_1)\models\psi$}}
\end{align*}
\caption{An informal scheme of the reductions and interpretations
	used in the proof of Theorem~\ref{thm:main}.}
\label{fig:scheme_main}
\end{figure}

\begin{proof}
We will show that if there exists a graph class $\cG$
satisfying all three properties stated above,
then we contradict the non-uniform ETH.
Fix $b \in \N$ (to be determined later from Lemma~\ref{lem:interpretin13regular}) 
and any sufficiently large $c \in \N$
such that~$c > 5b$. By (\ref{it:densunb}) and Definition~\ref{def:denselyunboundedII},
we have that for all $m \in \N$ there is $G'_m \in \cG$ such that $\tw(G'_m)\geq m^{5b}$
and $|V(G'_m)| < \cO \big(2^{m^{5b/c}} \big)$. 

By Proposition~\ref{thm:ReedWood}, 
the graph~$G'_m$ contains a subgraph $G_m \subseteq G'_m$ 
which is grid-like as $(G_m,\cP_m)$ of order $m^b$, for all sufficiently large~$m$.
Also $G_m\in\cG$ by (\ref{it:clsubgraph}).
We fix (one of) the $K_{m^b}$\,-minor in $I(\cP_m)$,
and denote by $\cV_m$ the partition of the vertex set of $I(\cP_m)$ 
into connected subgraphs that define this minor.
Furthermore, by Theorem~\ref{thm:strongedgecolouring},
there exists a strong edge colouring $\gamma_m\colon E(G_m) \to \{1, \ldots,22\}$ of~$G_m$.
Define an advice function $A$ that acquires the values
$A(m):=\langle G_m,\cP_m,\cV_m,\gamma_m\rangle$ 
(whenever $m$ is large enough for $G_m$ to be defined as above).
Since $c>5b$ and $|V(G_m)|<\cO\big(2^{m^{5b/c}}\big)$,
our advice function $A$ is sub-exponentially bounded;
$|A(m)|=\cO\big(|V(G)|^2\big)< \cO\big(2^{2m^{5b/c}}\big)$.

\smallskip
Now we get to the core of the proof (cf.~Figure~\ref{fig:scheme_main}).
Assume that we get an arbitrary graph $F$ and any $\MSO$ formula $\varphi$ as input.
We will show that the model-checking instance $F\models\varphi$ can be
solved in sub-exponential time wrt.\ $m=|V(F)|$
with help of our advice function $A$.
By Lemma~\ref{lem:interpretin13regular}, there is an interpretation $I_1$
such that there exists a $\{1,3\}$-regular graph $H$ and $H^{I_1}\simeq F$.
Moreover, since $I_1$ is efficient, we can compute $H$ efficiently
and $|V(H)|\leq m^b$ for a suitable fixed~$b$ and sufficiently large~$m$.
Then, we query the oracle advice value 
$A(m)=\langle G_m,\cP_m,\cV_m,\gamma_m\rangle$.
Since our advice $(G_m,\cP_m)$ is a grid-like graph of order $m^b$---%
i.e., its intersection graph $I(\cP_m)$ has a $K_{m^b}$\,-minor---%
\,$I(\cP_m)$ has a minor isomorphic to $H$, too.
But $H$ is $\{1,3\}$-regular and, in particular, has maximum degree three.
Hence there exists a subgraph $H_1\subseteq I(\cP_m)$ that is 
isomorphic to a subdivision of~$H$ 
(in other words, $H$ is a topological minor of~$I(\cP_m)$).
This subgraph $H_1$ can be straightforwardly computed from the advice $\cV_m$ over 
$(G_m,\cP_m)$ in polynomial time.

By Lemma~\ref{lem:keyinterpret} there is another efficient interpretation $I_2$
assigning to $H_1$ a labeling $\lambda_1$ such that
$(G_m,\lambda_1)^{I_2}\simeq H_1$.
This $\lambda_1$ can actually be computed very easily with help of the
advice $\gamma_m$ from $A(m)$ along the lines of the proof of
Lemma~\ref{lem:keyinterpret}, not even using the algorithmic
part of Theorem~\ref{thm:strongedgecolouring}.
Finally, we compute in polynomial time the formula
$\psi\equiv(\varphi^{I_1})^{I_2}$.
According to Lemma~\ref{lem:interpretin13regular}, 
$\psi$ is invariant under subdivisions of edges, and so 
$H \models \varphi^{I_1}$~$\iff$~$H_1 \models \varphi^{I_1}$.
Then, by the interpretation principle,
$F \models \varphi$~$\iff$~$H \models \varphi^{I_1}$~$\iff$~$H_1 \models
\varphi^{I_1}$~$\iff$~$(G_m,\lambda_1) \models \psi$.
The final task is to run the algorithm of (\ref{it:MSOinXP})
on the instance $(G_m,\lambda_1)\models\psi$.
The run-time is $|V(G_m)|^p$ for some $p$ depending only on~$\psi$,
i.e.\ only on $\varphi$.
Recall that $m=|V(F)|$ and $|V(G_m)|<\cO\big(2^{m^{5b/c}}\big)$.
Hence we get a solution to the model-checking instance $F\models\varphi$
in time $\cO\big(|V(G_m)|^{f(|\varphi|)}\big) <\cO\big(2^{f(|\varphi|)\cdot
	m^{5b/c}}\big)\in2^{\cO(m^{1-\varepsilon})}$
for any fixed $\varphi$, with a sub-exponentially bounded oracle advice
function $A$. 

In particular, if~$\varphi$ expresses the fact that a graph is 3-colourable
(Example~\ref{ex:3colorMso1}),
then this shows that $\text{\sc 3-Colourability} \in \subexpsubexp{m}$, 
contradicting non-uniform ETH.
\end{proof}

\begin{proposition}
Theorem~\ref{thm:main} remains valid even if (\ref{it:densunb}) is replaced
with ``the tree-width of $\cG$ is densely unbounded by 
$\log^{q\cdot\gamma}$ with gap degree $\gamma$''
for any $q>8$.
\end{proposition}
\begin{proof}[Proof sketch]
This follows from Definition~\ref{def:denselyunbounded}
and since Lemma~\ref{lem:interpretin13regular} works letting $b=2$
(cf.,~\cite{Ganianetal10}).
Combining with Proposition~\ref{thm:ReedWood}, we see that any exponent
$q>2\cdot4$ suffices for our arguments to work, modulo the gap degree.
\end{proof}

\section{Extending the Main Theorem}
\label{sec:tmt-s}

We can strengthen Theorem~\ref{thm:main} by showing that even every problem
in the {\em Polynomial-Time Hierarchy} (\PH)~\cite{Sto76} is in~$\subexpsubexp{n}$,
i.e., admits subexponential-sized circuits.
This stronger new conclusion comes at the price of a stricter assumption on the
graph class~$\cG$; we assume that the \MCa{\MSOL}{\cG^L} model-checking 
problem is in \XP\
for \emph{every} finite set of labels~$L$ such that~$|L| =\ca O(|\varphi|)$,
i.e., wrt.\ the formula size~$|\varphi|$ as a parameter
determining also the label set~$L$.
Note that in Theorem~\ref{thm:main}, $L$ was a fixed finite set of labels. 

We also study what happens if
we drop the condition that $\cG$ is \emph{densely} unbounded, and
only require that $\cG$ does not have poly-logarithmically bounded tree-width 
(i.e., there might be arbitrarily large gaps between the graphs
witnessing large tree-width in $\cG$).
Then we can show that all problems in the Polynomial-Time Hierarchy would 
admit \emph{robust simulations}~\cite{FS11} using subexponential-sized circuits.
% The notion of
% robust-simulations was introduced by Fortnow and Santhanam in a
% recent paper~\cite{FS11}.

\subsection{PH collapse result}

Our strategy to prove this result is as follows. We first define
a problem which we call $\SigmaCol k$ and show it to be complete
for~$\Sigma_k^p$, the {\em$k$-th level of \PH}. The problem $\SigmaCol k$
turns out to be expressible in \MSOL for each~$k$, though, the required 
set of labels~$L$ depends on $k$.
Now any language in \PH\ reduces to
$\SigmaCol k$ for some~$k$ and hence it is sufficient to
show that $\SigmaCol k \in \subexpsubexp{n}$ for all~$k$. We show
this by mimicking the proof of Theorem~\ref{thm:main}. We start
by defining the problem $\SigmaCol k$.
 
% \begin{definition}
% Let $G = (V, E)$ a graph.  
For a graph $G$ and a set $S \subseteq V(G)$,
a function $f\colon S \to \{1,2,3\}$ is called a {\em precolouring} 
of $G$ on $S$ iff the induced subgraph $G[S]$ is properly three-coloured.
For two precolourings $f_i\colon S_i \to \{1,2,3\}$, $i=1,2$,
with $S_1 \cap S_2 = \emptyset$, we let $f = f_1 \cup f_2$
be defined as $f\colon S_1 \cup S_2 \to \{1,2,3\}$ such that for
all $x \in S_1 \cup S_2$, $f(x) = f_i(x)$ iff $x \in S_i$.
% \end{definition}

\begin{definition}[Alternating colouring, and $\SigmaCol k$]
\label{def:alternating-colouring}
Let $G$ be a graph, $k$ an odd positive integer, 
$V_0, V_1, \ldots, V_{k} \subseteq V(G)$ be a partition of $V(G)$,
and $f_0\colon V_0 \to \{1,2,3\}$ be a precolouring of $G$ on $V_0$.
A \emph{$k$-alternating colouring for $(G, f_0, V_0, V_1, \ldots, V_{k})$}
is a function $f_1\colon V_1 \to \{1,2,3\}$ such that
\begin{enumerate}[i)]
\item $f_0 \cup f_1$ is a precolouring for $V_0 \cup V_1$; and
\item if $k > 1$, for all $f_2\colon V_2 \to \{1,2,3\}$ such that
	$f_0 \cup f_1 \cup f_2$ is a precolouring for $V_0 \cup V_1 \cup V_2$,
	there exists a $(k-2)$-alternating colouring for
	$(G, f_0', V_0', V_3, \ldots, V_{k})$,
	where $V_0' = V_0 \cup V_1 \cup V_2$ and $f_0' = f_0 \cup f_1 \cup f_2$.
\end{enumerate}
% \end{definition}
% \begin{definition}

For any odd $k \in \N$, the {\em problem $\SigmaCol k$} is defined as follows:
Given a graph $G$, a partition $V_0\cup V_1\cup \ldots\cup V_{k}=V(G)$, and
a precolouring $f_0\colon V_0 \to \{1,2,3\}$, decide whether
there is a $k$-alternating colouring for $(G, f_0, V_0, V_1, \ldots, V_{k})$.
\end{definition}

Recall that a \emph{polynomial-time many-one honest reduction} from~$L_1$ to~$L_2$ 
is a polynomial-time computable function~$f \colon \mathbf{N} \rightarrow \mathbf{N}$
such that $x \in L_1$ iff $f(x) \in L_2$ and~$|x|^{1/b} \le |f(x)| \le |x|^{b}$ for some
integer~$b>0$~\cite{DF03}.

Note that for $k = 1$ and $V_0 = \emptyset$, $V_1 = V$, the problem is 
the classical $\threecolor$ problem and hence complete for $\Sigma_{1}^p = \NP$.
More generally:
\begin{theorem}
\label{thm:sigmacol-hardness}
For each odd positive integer $k$, the $\SigmaCol k$ problem
is complete for $\Sigma_k^p$
under honest polynomial-time many-one reductions.
\end{theorem}

\begin{proof}
Containment follows from the existence of an alternating Turing machine that guesses the colouring of vertices
in the respective sets $V_i$.
For hardness, consider the problem $\SigmaSat k$ (also known as $\rm QSAT_k$)
which is the set of true quantified
Boolean formulas with $k-1$ quantifier alternations beginning with an
$\exists$-quantifier, such that the formulas are in CNF for odd~$k$ and
in DNF for even~$k$.
By~\cite{Sto76,Wra76}, for each $k \in \N$, $\SigmaSat k$ is complete for $\Sigma_k^p$
under honest polynomial-time many-one reductions.
We give a polynomial-time many-one reduction from $\SigmaSat k$ to
$\SigmaCol k$ by extending the standard reduction from
SAT to \threecolor.  Given an input
$\exists \tilde x^1 \forall \tilde x^2 \cdots \exists \tilde x^k
\varphi(\tilde x^1, \ldots, \tilde x^k)$ to $\SigmaSat k$,
where $\varphi$ is a Boolean formula in CNF and $(\tilde x^1, \ldots, \tilde x^k)$
is a partition of the variables in $\varphi$ such that a variable
in $\tilde x^i$ is existentially quantified if
$i$ is odd and universally otherwise, we create a graph $G = (V, E)$ as follows:
\begin{figure}[tbp]
\hfill \includegraphics{fig08.4} \hfill \includegraphics{fig08.5} \hfill{}
\caption{$\SigmaCol k$ reduction; left: variables; right: or-gadget\label{fig:3colred1}}
\end{figure}
\begin{itemize}
\item First, we create a triangle with distinct vertices $\oplus$ (``true''),
	$\ominus$ (``false''), and $\otimes$ (``forbid''), and 
\item for each variable $x$, we create an edge between two 
	distinct vertices $v_{x}$ and $v_{\bar x}$, and connect
	both vertices to $\otimes$.
	The result is depicted in Figure~\ref{fig:3colred1}.
\item For each CNF clause $\{l_1, l_2, \ldots, l_m\}$, we use $\ca O(m)$
	of the OR-gadgets depicted in Figure~\ref{fig:3colred1}.
	The output vertex of each OR-gadget is connected to $\otimes$.
	The output of the final OR-gadget for each clause is additionally
	connected to $\ominus$.
\item We let $V_0 = \{\ominus, \oplus, \otimes\}$ and 
	$f_0$ be defined as $f_0(\ominus) = 1$, $f_0(\oplus) = 2$, and $f_0(\otimes) = 3$.
\item For each $1 \le i \le k$, we let
	$V_i \supseteq \{\,v_{x}, v_{\bar x} \mid x \in \tilde x^i\,\}$,
	and additionally let $V_k$ contain all OR-gadgets.
\end{itemize}
It is not hard to see that this reduction takes polynomial time.
We induct over $k - l$ and show that for every even $0 \le l \le k-1$
 the following holds:
Let $\alpha$ be an assignment to the variables in
$\tilde x := \tilde x^1 \cup \cdots \cup \tilde x^l$.
Then $\exists \tilde x^{l+1} \forall \tilde x^{l+2} \cdots \exists \tilde x^k
\varphi(\alpha(\tilde x^1), \ldots, \alpha(\tilde x^l), \tilde x^{l+1},
\ldots, \tilde x^k) = 1$,
iff there is a $(k-l)$-alternating colouring for
$(G, f_0', V_0', V_{l+1}, \ldots, V_k)$,
where $V_0' = V_0 \cup \cdots \cup V_{l}$ and $f_0'\colon V_0' \to \{1,2,3\}$ with
$f_0'(v_{x}) = 1 + \alpha(x)$ and $f_0'(v_{\bar x}) = 2 - \alpha(x)$
for all variables $x \in \tilde x$.

The base case of induction is $l = k-1$.
Suppose $\exists \tilde x^{k} \varphi(\alpha(\tilde x^1), \ldots, \alpha(\tilde x^l),
\tilde x^k) = 1$. Then there is an assignment $\alpha'$ to the variables
of $\tilde x^k$, such that $\varphi(\alpha(\tilde x^1), \ldots, \alpha(\tilde x^l),
\alpha'(\tilde x^k)) = 1$.
We need to show that there is a $1$-alternating colouring
for $(G, V_0', f_0', V_1)$, i.e., a precolouring $f_1\colon V_1 \to \{1,2,3\}$,
such that $f_0' \cup f_1$ is a proper three-colouring of the graph.
Let for each $x \in \tilde x^k$,
$f_1(v_{x}) = 1 + \alpha'(x)$ and $f_1(v_{\bar x}) = 2 - \alpha'(x)$.
Then, using the same arguments as for the standard ${\rm SAT} \le_m \threecolor$
reduction, $f_0' \cup f_1$ is a three-colouring of the graph.
For the converse direction, suppose there is a precolouring $f_1\colon V_1 \to
\{1,2,3\}$ such that $f := f_0' \cup f_1$ is a three-colouring of the graph.
For $x \in \tilde x^k$, let $\alpha'(x) := f(v_{x}) - 1$.
Since for each variable $x$, the vertices $v_{x}$ and $v_{\bar x}$
are connected to $\otimes$, we know their colours are either $\ominus, \oplus$
or $\oplus, \ominus$, i.e., $\alpha'(x) \in \{0,1\}$.
Similarly, the output vertex of every OR-gadget is coloured
either $\ominus$ or $\oplus$.
In particular, the output vertex of the final OR-gadget for a clause
$\{l_1, \ldots, l_m\}$ is connected to both, $\ominus$ and $\otimes$, which
implies that it is coloured $\oplus$. 
Using a simple case distinction, we find that the output vertex of an OR-gadget
is forced to a colour $i$ if both inputs are connected to vertices coloured~$i$.
Therefore, if the final output vertex is coloured $\oplus$, there must be
be a $1 \le j \le m$ such that $f(v_{l_j}) = f(\oplus)$.
If this literal $l_j$ is, say, positive, i.e., $l_j = x$,
then either $x \in \tilde x$ and $\alpha(x) = f_0'(v_{l_j}) - 1 = 1$,
or $x \in \tilde x^k$ and $\alpha'(x) = f_1(v_{l_j}) - 1 = 1$.
In either case, the clause is satisfied.

For the induction step, let $0 \le l \le k - 3$ be even.
Suppose
$$\exists \tilde x^{l+1} \forall \tilde x^{l+2} \cdots \exists \tilde x^k
\varphi(\alpha(\tilde x^1), \ldots, \alpha(\tilde x^l), \tilde x^{l+1},
\ldots, \tilde x^k) = 1.$$
Then there is an assignment $\alpha_1$ to the variables of $\tilde
x^{l+1}$, such that for all assignments $\alpha_2$ to the variables of $\tilde
x^{l+2}$,
$$
\exists \tilde x^{l+3} \forall \tilde x^{l+4} \cdots \exists \tilde x^k
\varphi(\alpha(\tilde x^1), \ldots, \alpha(\tilde x^l),
\alpha_1(\tilde x^{l+1}),
\alpha_2(\tilde x^{l+2}),
\tilde x^{l+3}, \ldots, \tilde x^k) = 1.
$$
Let for each $x \in \tilde x^{l+1}$,
$f_1(v_{x}) = 1 + \alpha_1(x)$ and $f_1(v_{\bar x}) = 2 - \alpha_1(x)$.
Then $f_0' \cup f_1$ is a precolouring for $V_0' \cup \cdots \cup V_{l+1}$.
Furthermore, for all assignments $\alpha_2$ to the variables of $\tilde x^{l+2}$,
$f_0' \cup f_1 \cup f_2$,
where $f_2(v_{x}) = 1 + \alpha_2(x)$ and
$f_2(v_{\bar x}) = 2 - \alpha_2(x)$,
is a precolouring of $V_0' \cup V_{l+1} \cup V_{l+2}$, and therefore,
by the induction hypothesis for $l+2$,
there is a $(k-l-2)$-alternating colouring for
$(G, f_0' \cup f_1 \cup f_2, V_0' \cup V_{l+1} \cup V_{l+2}, V_{l+3}, \ldots,
V_{k})$.
Additionally, since all vertices in $V_{l+2}$ are connected to $\otimes$,
there is a one-to-one correspondence between those $f_2\colon V_{l+2} \to
\{1,2,3\}$, where $f_0' \cup f_1 \cup f_2$ is a precolouring of
$V_0' \cup V_{l+1} \cup V_{l+2}$, and the assignments $\alpha_2$.
As of Definition~\ref{def:alternating-colouring},
$f_1$ therefore satisfies the properties of a $(k-l)$-alternating colouring
for $(G, f_0', V_0', V_{l+1}, \ldots, V_k)$.

Conversely, suppose $f_1$ is a $(k-l)$-alternating colouring
for $(G, f_0', V_0', V_{l+1},$ $ \ldots, V_k)$ and 
consider an arbitrary $f_2\colon V_{l+2} \to \{1,2,3\}$
such that $f_0' \cup f_1 \cup f_2$ is
a precolouring for $V_0' \cup V_{l+1} \cup V_{l+2}$.
Then $(f_1 \cup f_2)(v) \in \{1,2\}$ for
every $v \in V_{l+1} \cup V_{l+2}$, since all of these vertices are
connected to $\otimes$ with $f_0'(\otimes) = 3$.
For $1 \le i \le 2$ and each $x \in V_{l+i}$,
let $\alpha_i(x) = f_i(x) - 1$.
By Definition~\ref{def:alternating-colouring}, there is a $(k-l-2)$-alternating
colouring for
$(G, f_0' \cup f_1 \cup f_2, V_0' \cup V_{l+1} \cup V_{l+2}, V_{l+3}, \ldots, V_k)$,
and hence, by the induction hypothesis,
$$
\exists \tilde x^{l+3} \cdots \exists \tilde x^k
\varphi(\alpha(\tilde x^1), \ldots, \alpha(\tilde x^l),
\alpha_1(\tilde x^{l+1}),
\alpha_2(\tilde x^{l+2}),
\tilde x^{l+3}, \ldots, \tilde x^k) = 1.
$$
Again, there is a one-to-one correspondence between
assignments $\alpha_2$ to the variables in $\tilde x^{l+2}$
and functions $f_2\colon V_{l+2} \to \{1,2,3\}$ such that
$f_0' \cup f_1 \cup f_2$ is a precolouring for $V_0' \cup V_{l+1} \cup V_{l+2}$,
because all vertices in $V_{l+2}$ are connected to the vertex $\otimes$.
Therefore, since $f_2$ was arbitrary, the formula holds for all assignments
of the variables in $\tilde x^{l+2}$, which implies
$$
\exists \tilde x^{l+1} \forall \tilde x^{l+1}\cdots \exists \tilde x^k
\varphi(\alpha(\tilde x^1), \ldots, \alpha(\tilde x^l),
\tilde x^{l+1}, \ldots, \tilde x^k) = 1,
$$
which concludes the proof.
\end{proof}

\begin{lemma}
\label{lem:SigmaColMSO}
$\SigmaCol k$ can be expressed in $\rm MSO_1$-L for odd~$k$ with $|L|=k+3$.
\end{lemma}
\begin{proof}
We assume an input $(G, f_0, V_0, V_1, \ldots, V_{k})$
is encoded as the graph $G$ together with labels $V_1,\dots,V_k$
determining the corresponding sets of the vertex partition.
Let the three colours be ``Red'', ``Green'', and ``Blue''.
Then we use three additional vertex labels $R_0, G_0, B_0$ to encode the
values of the precolouring $f_0$ on $V_0$, which is part of the input.
For each $0 \le i \le k$, we define a routine MSO-formula $\Precol_i$ that expresses:
\begin{enumerate}
\item $R_i,G_i,B_i$ is a partition of $V_i$, and
\item $(\bigcup_{0 \le j \le i} R_j, \bigcup_{0 \le j \le i} G_j, 
	\bigcup_{0 \le j \le i} B_j)$ \smallskip
	is a proper $3$-colouring of the induced subgraph
	$G[V_0 \cup \cdots \cup V_i]$.
\end{enumerate}
Here $R_j,G_j,B_j$ are implicit free variables for $0<j\leq i$.

Then the formula for $\SigmaCol k$ is constructed as follows:
\begin{multline*}
\exists R_1,G_1,B_1 \mbox{``}\subseteq V_1\mbox{''} \big[\Precol_1 \wedge
	\big(\forall R_2,G_2,B_2 \mbox{``}\subseteq V_2\mbox{''} (\Precol_2 \to
\\	(\exists R_3,G_3,B_3 \mbox{``}\subseteq V_3\mbox{''} \cdots
	(\exists R_k,G_k,B_k  \mbox{``}\subseteq V_k\mbox{''}
	 \Precol_k)\cdots))\big)\big]
\end{multline*}
Note that the formula only depends on $k$.
\end{proof}

% \begin{remark}
% With a bit more care, one can prove .... with $\log k$ labels.....
% \end{remark}

\begin{theorem}\label{thm:PH_subexp}
Unless $\PH \subseteq
 \subexpsubexp{n}$, there exists no graph class $\cG$ satisfying all three properties
\begin{enumerate}[a)]
\item
$\cG$ is closed under taking subgraphs,
\item
the tree-width of $\cG$ is densely unbounded poly-logarithmically,
\item
the \MCa{\MSOL}{\cG^L} model-checking problem of a sentence $\varphi$ 
is in \XP for all label sets $L$ such that $|L|=\ca O(|\varphi|)$,
i.e., one can test whether $G\models\varphi$ where a graph
$G\in\cG$ is vertex-labeled with $\ca O(|\varphi|)$ labels, 
in time $\cO\big(|V(G)|^{f(|\varphi|)}\big)$ for some computable function~$f$.
\end{enumerate}
\end{theorem}

\begin{proof}%[Proof sketch]
Let $\ca L \in \Sigma_k^p$.  Let~$F$ be a polynomial-time many-one 
reduction from~$\ca L$ to $\SigmaCol k$ that runs in time~$\ca O(|x|^d)$
where $d\in\N$ is constant.  
On input $x$ of $\ca L$, we use~$F$ to map it to an instance~$F(x)$ of $\SigmaCol k$. 
Note that $|F(x)| \le |x|^d$. 

Now, $\SigmaCol k$ can be expressed in \MSOL with $k+3$ labels by
Lemma~\ref{lem:SigmaColMSO}, say by a sentence $\sigma_k$
over $L$-vertex-labeled graphs.
Choosing $c > 5bd$ (to compensate for the small increase in size of $x$),
we then continue exactly as in the proof of Theorem~\ref{thm:main}
with the advice~$A(m)$ where $m = |x|^{d}$:
\begin{itemize}
\item Namely,
by Lemma~\ref{lem:interpretin13regular}, we in time~$|F(x)|^b$ map~$F(x)$ into a
$\{1,3\}$-regular $L$-graph~$H$ such that~$F(x)$ has a $k$-alternating
colouring iff~$H \models \varrho$, where~$\varrho=\sigma_k^{I_1}$ 
comes from the first interpretation step, keeping its $L$-labels.
\item
In the second interpretation step, we construct (with help of $A(m)$)
a corresponding $L'$-labeling of $G_m$ where $|L'|=|L|+47=k+50$,
and the interpreted formula $\psi=\varrho^{I_2}=(\sigma_k^{I_1})^{I_2}$.
The rest follows exactly as in the former proof.
\end{itemize}
With $m = |x|^{d}$ and $c > 5bd$, the advice $A(m)$ is sub-exponentially
bounded in $|x|$, as $2^{\cO(|x|^{1-\varepsilon})}$ for every fixed $k$
and some $\varepsilon>0$ depending on~$k$ only, 
and so is the total running time.
\end{proof}

\subsection{Robustly-often Simulations}

% If we drop the condition that $\cG$ is \emph{densely} unbounded and
% only require that $\cG$ does not poly-logarithmically bounded tree-width (i.e., there
% might be arbitrary large gaps between such graphs in $\cG$),
% then we can show that all problems in the Polynomial-Time Hierarchy admit \emph{robust
% simulations} using subexponential-sized circuits.  

The notion of {\em robust simulations} was introduced by 
Fortnow and Santhanam in a recent paper~\cite{FS11}.
Given a language~$L$ and a complexity class~$\cC$, it is $L \roin \cC$ 
(robustly-often in~$\cC$), if there is a language $A$ in $\cC$ such that
for every $j \in \N$ there are infinitely many $m$ such that $A$ 
and $L$ agree on all input lengths between $m$ and $m^j$. Intuitively,
if $L \roin \cC$, then there exists a $\cC$-algorithm that solves infinitely
many ``polynomially-wide patches'' of input lengths of instances of~$L$.
% for every~$j \in \N$.
For a formal definition of the concept of \emph{robustly-often simulations}, 
we need some more definitions from~\cite{FS11}:

\begin{definition}\label{defn:robustset}
% Let $S$ be a subset of the positive integers.  Then
Let $L$ be a language and $\cC$ be a complexity class.
\begin{itemize}
\item
A set of positive integers $S$ is {\em robust} if for each positive integer $j$ there
is a positive integer $m \ge 2$ such that $\{m,m+1,\dots,m^j\} \subseteq S$.
% is a positive integer $m \ge 2$ such that $n \in S$ for all $m \le n \le m^j$.
%The canonical refinement
%$S_j$ of $S$ at level~$j$ is defined for every integer~$j > 0$: $m \in S_j$
%iff $m \in S$ and $n \in S$ for all $m \le n \le m^j$.
\item
We say the language $L$ is in $\cC$ on $S\subseteq \N$, 
if there is a language $L' \in \cC$ such that $L_n = L_n'$ for all $n \in S$ 
(where $L_n$ and $L_n'$ are the sets of words of length $n$ in $L$ and $L'$, respectively).
\item
We say that $L \roin \cC$ if there is a robust set $S$ such that $L \in \cC$ on $S$. 
In such a case we say that there is a {\em robustly-often simulation} of $L$ in $\cC$.
\end{itemize}
\end{definition}

Given a function $s\colon \N \to \N$, we denote by $\SIZE(s)$ the class of Boolean functions
$f = \{f_n\}$ such that for each $n$, $f_n$ has Boolean circuits of size $\cO(s(n))$.

\begin{lemma}
\label{lem:col-np-collapse}
If $\threecolor \roin \SIZE(2^{n^{1/c}})$ for all $c > 1$, then
$\NP \rosubseteq \SIZE(2^{n^{1/d}})$ for all $d > 1$.
\end{lemma}
\begin{proof}
We essentially follow the proof of Lemma~1 in~\cite{FS11}.
Fix $d > 1$. Let $L \in \NP$ and let, w.l.o.g., $f$ be an honest polynomial-time many-one
reduction from $L$ to $\threecolor$ such that for any word $x$ we have
that $|x|^{1/b} \le |f(x)| \le |x|^b$ (the first inequality can
always be achieved by padding).  Choose $c = bd$ and let $S$ be a
robust set such that $\threecolor \in \SIZE(2^{n^{1/c}})$ on~$S$.  
Define~$S'$ as follows: $n' \in S'$ iff %$n' \in S$ and 
for all words~$x$ of length~$n'$, we have~$|f(x)| \in S$.

We claim that~$S'$ is robust. The robustness of~$S$ is equivalent to
saying that for each positive integer~$j$ there exists an integer~$m_j$
such that~$n \in S$ for all~$m_j^{1/j} \leq n \leq m_j^j$. To show
that~$S'$ is robust, too, we need to exhibit for each positive~$j$ an
integer~$q_j$ such that~$p \in S'$ for all~$q_j^{1/j} \leq p \leq q_j^j$.
Fix an integer~$i$ and choose~$q_i = m_{bi}$. Now~$n \in S$ for all
$q_{i}^{1/bi} \le n \le q_{i}^{bi}$. If~$q_{i}^{1/i} \le p \le
 q_{i}^{i}$ and~$x$ is a word of length~$p$, then $q_{i}^{1/bi} \le
 |f(x)| \le q_{i}^{bi}$ and hence~$|f(x)| \in S$. By the definition
 of~$S'$, all integers between~$q_{i}^{1/i}$ and~$q_{i}^i$ are in~$S'$, proving 
that~$S'$ is robust. 

Finally, we show that $L \in \SIZE(2^{n^{1/d}})$ on $S'$.  Let $x$ be 
a word such that $|x| \in S'$.
Since $\threecolor \in \SIZE(2^{n^{1/c}})$ on~$S$, there is a Boolean circuit of size
$\cO(2^{|f(x)|^{1/c}})$ that decides membership of $f(x)$ in \threecolor.
This circuit has size $\cO(2^{|x|^{b/c}}) = \cO(2^{|x|^{1/d}})$, which
concludes the proof. 
\end{proof}

\begin{corollary} \label{cor:SigmakCol_PH}
If $\SigmaCol k \roin \SIZE(2^{n^{1/c}})$ for all odd $k \in \N$ and all
$c > 1$, then $\PH \rosubseteq \SIZE(2^{n^{1/d}})$ for all $d > 1$.
\end{corollary}
\begin{proof}
The proof is almost identical to the proof of Lemma~\ref{lem:col-np-collapse}.
Fix $d > 1$.
Let $L \in \PH$.  Then $L \in \Sigma_k^p$ for some $k \in \N$.
We can w.l.o.g.\ assume $k$ is odd.  By
Theorem~\ref{thm:sigmacol-hardness}, $\SigmaCol k$ is hard for $\Sigma_k^p$,
i.e., there is an polynomial-time many-one
reduction from $L$ to $\SigmaCol k$, which can be made honest by padding extra vertices in $V_0$.
We can now continue exactly as in the proof
of Lemma~\ref{lem:col-np-collapse}.
\end{proof}

In accordance with Definition~\ref{def:denselyunboundedII},
we say that the tree-width of $\cG$ is
{\em unbounded poly-logarithmically} (i.e., dropping the
``dense'' property and allowing arbitrarily large gaps between witnesses);
if, for all $c > 1$, there are infinitely many $m \in \N$
such that there exists $G\in\cG$ whose tree-width is $\tw(G)\geq m$ and size
$|V(G)| < \cO\big(2^{m^{1/c}}\big)$.
% \begin{definition}
% For a graph class $\cG$, we say that the tree-width of $\cG$ is
% {\em unbounded poly-logarithmically} if for all $c > 0$ there are
% infinitely many $m \in \N$, such that there is a graph $G\in\cG$
% whose tree-width is $\tw(G)\geq m$ and size
% $|V(G)| < \cO\big(2^{m^{1/c}}\big)$.
% \end{definition}
% That is, compared to Definition~\ref{def:denselyunboundedII} we drop the
% ``dense'' property and allow arbitrarily large gaps between witnesses.  
We then have the following collapse result under robustly-often simulations.

\begin{theorem}
Unless $\PH \rosubseteq \SIZE(2^{n^{1/d}})$ for any $d > 1$,
there exists no graph class $\cG$ satisfying all the three properties
\begin{enumerate}[a)]
\item
$\cG$ is closed under taking subgraphs,
\item
the tree-width of $\cG$ is unbounded poly-logarithmically,
\item
the \MCa{\MSOL}{\cG^L} model-checking problem of a sentence $\varphi$ 
is in \XP for all label sets $L$ such that $|L|=\ca O(|\varphi|)$,
i.e., one can test whether $G\models\varphi$ where a graph
$G\in\cG$ is vertex-labeled with $\ca O(|\varphi|)$ labels, 
in time $\cO\big(|V(G)|^{f(|\varphi|)}\big)$ for some computable function~$f$.

\end{enumerate}
\end{theorem}
\begin{proof}
By Corollary~\ref{cor:SigmakCol_PH}, we only need to show that if there
exists a graph class satisfying all the three properties mentioned above,
then we have \mbox{$\SigmaCol k \roin \SIZE(2^{n^{1/c}})$} for any $c > 1$
and any odd integer~$k$. Fix~$c > 1$, and let $k$ be an odd integer.
For $j\in\N$, call $m\in\N\>$ $j$-good if there exists
$G\in\cG$ whose tree-width is $\tw(G)\geq m^{ja}$ and size
$|V(G)| < \cO\big(2^{m^{1/2c}}\big) =
 \cO\big(2^{(m^{ja})^{1/2jac}}\big)$, where $a$ is any constant determined
later.
Clearly, there are infinitely many $j$-good integers $m$ for each $j\in\N$.

We set $M_c(m,j)=\{m,m+1,\dots,m^{j}\}$, and
$S_c=\bigcup\, \{\,M_c(m,j): m$ is $j$-good$\,\}$.
Then $S_c$ is robust by the definition.
The point is that the following holds from Lemma~\ref{lem:SigmaColMSO} and
the fine details of the proof of Theorem~\ref{thm:main}
(choosing $a=5b$ there):
For any $n$-vertex instance of $\SigmaCol k$
such that $n\in M_c(m,j)$ and $m$ is $j$-good, this instance can be 
solved---using the assumed algorithm of (c)---in time
$\cO\big(2^{m^{1/2c}\cdot g(k)}\big)<
 \cO\big(2^{m^{1/c}}\big)\leq \cO\big(2^{n^{1/c}}\big)$
with an advice of size bounded by the same function.
Hence, indeed, \mbox{$\SigmaCol k \roin \SIZE(2^{n^{1/c}})$}.
\end{proof}
%%%% The old proof made no sense...
%
% Since the class~$\cG$ has poly-logarithmically unbounded tree-width,
% there exists an infinite set~$S(c)$ of positive integers 
% such that~$n \in S(c)$ implies that there exists~$G \in \cG$ with
% $\tw(G) \ge n$ and~$|V(G)| \le 2^{n^{1/c}}$. The proof of 
% Theorem~\ref{thm:main} then shows that for every
% instance~$x$ of $\SigmaCol k$ of length~$n \in S(c)$ can be solved
% in time~$2^{n^{1/c}}$ with subexponential advice. 
% 
% We want to show that~$S(c)$ is robust. To do this, we need to 
% show that for every integer~$j$, there exists~$m \in S(c)$
% such that~$n \in S(c)$ for all~$m \le n \le m^j$. Fix~$j >0$
% and let~$p$ be the smallest integer such that~$p \in S(c)$ and~$1 < p^{1/5j}$.
% Such a choice is possible as~$S(c)$ is infinite. Let~$G \in \cG$
% be such that~$\tw(G) \ge p$ and~$|V(G)| \le 2^{p^{1/c}}$. By 
% Proposition~\ref{thm:ReedWood}, $G$ contains a grid-like graph~$H$
% of order~$p^{1/5}$ and as~$\cG$ is closed under taking subgraphs~$H \in
%  \cG$. Now the intersection graph of the set of paths in~$H$ contains
% a minor of order~$p^{1/5}$. Hence any instance~$x$ of~$\SigmaCol k$ 
% such that~$p^{1/5k} \le |x| \le p^{1/5}$ can be decided in 
% subexponential time using the advice about~$H$, as was done 
% in the proof of Theorem~\ref{thm:main}. Thus~$S(c)$ contains
% all integers between~$p^{1/5k}$ and~$p^{1/5}$ and the map~$j \mapsto
%  p^{1/5}$ shows that~$S(c)$ is robust.

\section{Implications for Directed Width Measures}
\label{sec:directed}
\label{sec:idwm}

In this section, we briefly foray into the area of digraph width measures
and discuss, in particular, the implications of the results
in the previous section. 
This part follows up on~\cite{Ganianetal10}.

An important goal in the design of a ``good'' width measure is for it
to satisfy two seemingly contradictory requirements:
\begin{enumerate}[I)]
\item a large class 
of problems must be efficiently solvable on the graphs of bounded width; and
\item the class of the graphs of bounded width should have 
a nice, reasonably rich and natural structure.
\end{enumerate}
In contrast to the undirected graph case, where e.g.\ tree-width has become
a true success story, this effort has largely
failed for digraph width measures. 
A partial answer for the reasons of this failure 
was provided in~\cite{Ganianetal10} where it was shown that any digraph width 
measure that is different from the undirected tree-width and monotone
under directed topological minors is not algorithmically powerful. 
The phrase ``different from tree-width'' is defined by the property
that there exists a constant~$c\in\N$ such that the class of the underlying 
undirected graphs of digraphs of width at most~$c$ has unbounded tree-width. 
Algorithmic ``powerfulness'' has been defined as the property
of admitting \XP algorithms (wrt.\ the width as parameter) for all problems in \MSO.

We improve upon this result by showing that even 
if the digraph width measure is monotone just under subdigraphs,
and the underlying undirected graphs corresponding to digraphs 
of bounded width have poly-logarithmically unbounded tree-width, 
then the width measure is not algorithmically powerful. 
First note that we relax \emph{unbounded} tree-width by 
\emph{poly-logarithmically unbounded} tree-width.
This is a somehow stronger assumption, and the strengthening is unavoidable
due to a negative example shown in~\cite{Ganianetal10}.
 
Secondly, we require the directed width measure to be closed under 
\emph{subdigraphs} and not \emph{directed topological minors} as in~\cite{Ganianetal10}; 
which is, on the other hand, a much weaker requirement. 
Thirdly, our interpretation of algorithmic
powerfulness is, now, that all problems in $\MSOL$ can be solved
on {\em$L$-vertex-labeled graphs} in \XP-time wrt.\ the width {and} formula size as parameters. 
This again is a dilution of the notion of algorithmic power 
as defined in~\cite{Ganianetal10}, 
where only plain \MSO over unlabeled digraphs has been exploited.

We start by defining what it means for a digraph width measure 
to have poly-logarithmically unbounded tree-width. 
We shortly denote by $U(D)$ the underlying undirected graph of a
digraph~$D$.
\begin{definition}
A directed width measure~$\delta$ {\em largely surpasses tree-width}
if there exists $d \in \N$ such that the tree-width of the 
undirected graph class $\{\,U(D): \delta(D)\leq d\,\}$ is densely unbounded
poly-logarithmically.
\end{definition}

Then the main result of this section reads:
\begin{theorem}
\label{thm:maindir}
Let $L$ be a finite set of labels, $|L|\geq47$.
Unless the non-uniform Exponential-Time Hypothesis fails, there exists no
directed width measure $\delta$ satisfying all three properties:
\begin{enumerate}[a)]
\item
$\delta$ is monotone under taking subdigraphs;
\item
$\delta$ largely surpasses the tree-width of underlying undirected graphs; and
\item \label{it:333}
for all $L$-vertex-labeled digraphs~$D$ and all
sentences~$\varphi \in \MSOL$, the problem of deciding whether $D \models \varphi$
is solvable in time~$\ca O(|D|^{f(\delta(D), |\varphi|)})$ for some computable~$f$.

\end{enumerate}
\end{theorem}
\begin{proof}
Assume that there exists a directed width measure~$\delta$ satisfying all
the three properties stated in the theorem. Since~$\delta$ largely
surpasses tree-width, there exists a constant~$d \in \N$ such
that the tree-width of the undirected graph class $\cG := \{\,U(D): \delta(D) \leq d\,\}$
is densely unbounded poly-logarithmically. Since~$\delta$ is monotone
under taking subdigraphs, the class~$\cG$ is closed on subgraphs.
Consider a formula $\varphi \in \MSOL$ on undirected $L$-vertex-labeled 
graph~$G$. If we construct a formula $\varphi'$ for $L$-vertex-labeled
digraphs by replacing every occurrence of the predicate $\adj(x,y)$ in~$\varphi$ 
with $\arc(x,y) \vee \arc(y,x)$, then $G \models \varphi$ iff for every 
orientation~$D$ of~$G$ it holds that~$D \models \varphi'$. 

To complete the proof, given any undirected graph~$F$ on $m$~vertices
and an \MSOL formula~$\varphi$, we use an advice function~$A(m) := \langle
 D_m, {\cP}_m, {\cV}_m, \gamma_m\rangle$ analogous to that used 
in Theorem~\ref{thm:main} to obtain a \emph{digraph}~$D_m$ 
such that~$\delta(D_m) \leq d$ and~$\big(U(D_m),{\cP}_m\big)$ is grid-like of order~$m$.
Note that for digraphs of constant $\delta$-width, the algorithm
guaranteed by condition~(\ref{it:333}) 
runs in \XP-time wrt.\ the size of the formula as parameter. 
We proceed as in the proof of Theorem~\ref{thm:main} to decide whether~$F \models \varphi$
in time~$2^{\cO(m^{1-\varepsilon})}$ using the sub-exponentially bounded oracle advice function~$A$. 
This again shows, in particular, that 
$\text{\sc 3-Colourability} \in \subexpsubexp{m}$, refuting non-uniform ETH.
\end{proof}

\section{Concluding Remarks}
\label{sec:concluding}

Our paper contributes to Kreutzer and Tazari's impressive results in this area.
Our proof is shorter and holds for $\MSOL$ logic instead of $\MSOii$ at
the price of a stronger assumption in computational complexity.
The expressive power of $\MSOii$ over graphs with labels from a set~$L$ and 
$\MSO$ with the same label set is huge---for instance, the latter is 
not able to express some natural graph problems like Hamiltonian cycle.
However, one cannot directly compare the expressive power of bare \MSOii
without labels and \MSOL over graphs with vertex labels from~$L$,
as there are problems which can be expressed in \MSOii but not 
in \MSOL and vice versa.
We have proved that it is not possible to efficiently process
latter \MSOL on graph classes with ``very'' 
unbounded tree-width which are subgraph-closed.

Besides the implications for digraph width measures discussed in Section \ref{sec:idwm}, 
there is also an implication for another width measure---clique-width.
Clique-width \cite{CMR00} (as well as rank-width) 
is a graph parameter which allows efficient ($\FPT$ time) 
model-checking of all $\MSOL$ formulas, however it has received 
some criticism for not having nice structural properties such as 
being monotone under taking subgraphs.
Our results indicate that it is unlikely any parameter exists with the 
desirable properties of clique-width which is monotone under taking subgraphs.

Finally, let us briefly mention the possibility of extending 
Theorem~\ref{thm:main} to unlabeled graphs, i.e., using plain \MSO
over $\cG$ in Theorem~\ref{thm:main}\,(\ref{it:MSOinXP}).
It is not known whether there exists any natural and nontrivial graph class 
where unlabeled $\MSO$ is efficiently solvable and yet $\MSOL$ 
model-checking is hard. Such a graph class would necessarily contain 
graphs of unbounded clique-width (since otherwise $\MSOL$ could be 
efficiently model-checked) and yet with sufficient structure to 
allow efficient model-checking of bare $\MSO$.
This indicates that such an extension is probably true.
For getting this ``unlabeled'' extension of Theorem~\ref{thm:main}
it would actually suffice to have an excluded grid theorem for graph tree-width
with a polynomial gap between the grid size and tree-width,
but that seems like a very difficult task at this moment.

\subsection*{Acknowledgements}

We thank Eric Allender for pointing out reference~\cite{FS11} and
Felix Reidl for useful discussions on robust simulations.

\bibliographystyle{abbrv}
\bibliography{gtbib}

\begin{thebibliography}{10}

\bibitem{CSH08}
V.~Chandrasekaran, N.~Srebro, and P.~Harsha.
\newblock Complexity of inference in graphical models.
\newblock In {\em UAI'08}, pages 70--78, 2008.

\bibitem{Cou90}
B.~Courcelle.
\newblock The monadic second order logic of graphs {I}: Recognizable sets of
  finite graphs.
\newblock {\em Inform. and Comput.}, 85:12--75, 1990.

\bibitem{CE12}
B.~Courcelle and J.~Engelfriet.
\newblock {\em Graph Structure and Monadic Second-Order Logic: A Language
  Theoretic Approach}.
\newblock Number 138 in Encyclopedia of Mathematics and its Applications.
  Cambridge University Press, June 2012.

\bibitem{CMR00}
B.~Courcelle, J.~A. Makowsky, and U.~Rotics.
\newblock Linear time solvable optimization problems on graphs of bounded
  clique-width.
\newblock {\em Theory Comput. Syst.}, 33(2):125--150, 2000.

\bibitem{Cra06}
D.~W. Cranston.
\newblock Strong edge-coloring of graphs with maximum degree 4 using 22 colors.
\newblock {\em Discrete Math.}, 306(21):2772--2778, 2006.

\bibitem{DF03}
R.~Downey and L.~Fortnow.
\newblock Uniformly hard languages.
\newblock {\em Theoret. Comput. Sci.}, 2(298):303--315, 2003.

\bibitem{EF99}
H.-D. Ebbinghaus and J.~Flum.
\newblock {\em Finite Model Theory}.
\newblock Springer, 1999.

\bibitem{fg06}
J.~Flum and M.~Grohe.
\newblock {\em Parameterized Complexity Theory}.
\newblock Springer, 2006.

\bibitem{FS11}
L.~Fortnow and R.~Santhanam.
\newblock Robust simulations and significant separations.
\newblock In {\em ICALP '11}, volume 6755 of {\em LNCS}, pages 569--580.
  Springer, 2011.

\bibitem{Ganianetal10}
R.~Ganian, P.~Hlin\v{e}n\'y, J.~Kneis, D.~Meister, J.~Obdr\v{z}\'alek,
  P.~Rossmanith, and S.~Sikdar.
\newblock Are there any good digraph width measures?
\newblock In {\em IPEC'10}, volume 6478 of {\em LNCS}, pages 135--146.
  Springer, 2010.

\bibitem{Gro08}
M.~Grohe.
\newblock Logic, graphs, and algorithms.
\newblock In {\em Logic and Automata: History and Perspectives}, pages
  357--422. Amsterdam University Press, 2008.

\bibitem{IPZ01}
R.~Impagliazzo, R.~Paturi, and F.~Zane.
\newblock Which problems have strongly exponential complexity?
\newblock {\em J. Comput. System Sci.}, 63(4):512--530, 2001.

\bibitem{Kre09CSL}
S.~Kreutzer.
\newblock On the parameterised intractability of monadic second-order logic.
\newblock In {\em CSL'09}, volume 5771 of {\em LNCS}, pages 348--363. Springer,
  2009.

\bibitem{Kre11}
S.~Kreutzer.
\newblock Algorithmic meta-theorems.
\newblock In {\em Finite and Algorithmic Model Theory}, number 379 in London
  Mathematical Society Lecture Notes. Cambridge University Press, 2011.

\bibitem{KT10b}
S.~Kreutzer and S.~Tazari.
\newblock Lower bounds for the complexity of monadic second-order logic.
\newblock In {\em LICS'10}, pages 189--198. IEEE, 2010.

\bibitem{KT10a}
S.~Kreutzer and S.~Tazari.
\newblock On brambles, grid-like minors, and parameterized intractability of
  monadic second-order logic.
\newblock In {\em SODA'10}, pages 354--364. SIAM, 2010.

\bibitem{mm03}
J.~A. Makowsky and J.~Mari{\~n}o.
\newblock Tree-width and the monadic quantifier hierarchy.
\newblock {\em Theoret. Comput. Sci.}, 303(1):157--170, 2003.

\bibitem{Rab64}
M.~O. Rabin.
\newblock A simple method for undecidability proofs and some applications.
\newblock In Y.~Bar-Hillel, editor, {\em Logic, Methodology and Philosophy of
  Sciences}, volume~1, pages 58--68. North-Holland, Amsterdam, 1964.

\bibitem{RW08}
B.~Reed and D.~Wood.
\newblock Polynomial treewidth forces a large grid-like-minor.
\newblock Technical Report abs/0809.0724, CoRR, 2008.

\bibitem{Sto76}
L.~Stockmeyer.
\newblock The polynomial-time hierarchy.
\newblock {\em Theoret. Comput. Sci.}, 3(1):1--22, 1976.

\bibitem{Wra76}
C.~Wrathall.
\newblock Complete sets and the polynomial-time hierarchy.
\newblock {\em Theoret. Comput. Sci.}, 3(1):22--33, 1976.

\end{thebibliography}
\end{document}